\documentclass[copyright,creativecommons]{eptcs}
\providecommand{\event}{WRS 2009} 
\providecommand{\volume}{??}
\providecommand{\anno}{20??}
\providecommand{\firstpage}{1}
\providecommand{\eid}{??}
\usepackage{amssymb,amsmath,amsthm,latexsym}
\usepackage{amsfonts}
\usepackage{stmaryrd}
\usepackage{breakurl}

\newcommand{\dbi}[1]{\,\underline{#1}}
\newcommand{\lvec}{\langle}
\newcommand{\rvec}{\rangle}

\newcommand{\ty}[2]{{#1}\!\to\!{#2}}
\newcommand{\ity}[2]{{#1}\land{#2}}
\newcommand{\app}[2]{({#1}\;{#2})}

\newcommand{\tyj}[3]{{#1}\!:\!\lvec{#2}\vdash{#3}\rvec}
\newcommand{\tyjs}[4]{{#1}\!:\!\lvec{#2}\vdash_{\scriptscriptstyle #3}{#4}\rvec}

\newcommand{\cty}[2]{{#1}\!\Rightarrow\!{#2}}

\newtheorem{theorem}{Theorem}
\newtheorem{definition}{Definition}
\newtheorem{example}{Example}
\newtheorem{lemma}{Lemma}
\newtheorem{coro}{Corollary}

\title{Principal Typings in a Restricted Intersection Type System for Beta Normal Forms with de Bruijn Indices}
\author{Daniel Ventura \qquad Mauricio Ayala-Rinc\'on 
\institute{Grupo de Teoria da Computa\c{c}\~ao, Dep. de Matem\'atica\\ Universidade de Bras\'{\i}lia\thanks{First (second) author supported by a CNPq PhD scholarship  (grant) at the Universidade de Bras\'ilia. Work supported by the Funda\c c\~ao de Apoio \`a Pesquisa do Distrito Federal [FAPDF 8-004/2007]}\\ Bras\'{\i}lia D.F., Brasil}
\email{[ventura,ayala]@mat.unb.br}
\and
Fairouz Kamareddine
\institute{School of Mathematical and Computer Sciences\\ 
Heriot-Watt University\\ Edinburgh, Scotland}
\email{fairouz@macs.hw.ac.uk}
}

\def\titlerunning{PT in a Restricted Intersection Type System for Beta Normal Forms with dB Indices}
\def\authorrunning{D. Ventura \& M. Ayala-Rinc\'on \& F. Kamareddine}

\begin{document}

\maketitle

\begin{abstract}
The $\lambda$-calculus with de Bruijn indices assembles each
$\alpha$-class of $\lambda$-terms in a unique term, using indices
instead of variable names. Intersection types provide finitary type
polymorphism and can characterise normalisable $\lambda$-terms
through the property that a term
is normalisable if and only if it is typeable. To be closer to
computations and to simplify the formalisation of the atomic
operations involved in $\beta$-contractions, several calculi of
explicit substitution were developed mostly with de Bruijn
indices. Versions of explicit substitutions calculi  without types and
with simple type systems are well investigated in contrast to versions
with more elaborate type systems such as intersection types. In a
previous work, we introduced a de Bruijn version of the
$\lambda$-calculus with an intersection type system and proved that it preserves subject reduction, a basic property of type systems. In this paper a version with de Bruijn indices of an intersection type system originally introduced to characterise principal typings for $\beta$-normal forms is presented. We present the characterisation in this new system and the corresponding versions for the type inference and the reconstruction of normal forms from principal typings algorithms. We briefly discuss the failure of the subject reduction property and some possible solutions for it. 
\end{abstract}

\section{Introduction}

The $\lambda$-calculus \`a la de Bruijn \cite{dB72} was introduced by the
Dutch mathematician  N.G.\ de Bruijn in the context of the project Automath \cite{NGDV} and has been adopted for several calculi of explicit substitutions ever since
(e.g. \cite{dB78,ACCL91,fayo}). Term variables in the $\lambda$-calculus \`a la de Bruijn are represented by indices instead of names, assembling each $\alpha$-class of terms in the
$\lambda$-calculus \cite{Bar1984} in a unique term with de Bruijn
indices, thus turning it more {\it ``machine-friendly''} than its
counterpart. Calculi with de Bruijn indices have been investigated for
both type free and simply typed versions. However, to the best of our
knowledge, apart from \cite{DBLP:journals/cj/KamareddineR02}, there is no work on using de Bruijn indices with more elaborate
type systems such as intersection type systems.  

Intersection types were introduced to provide a characterisation of the strongly normalising 
$\lambda$-terms \cite{CDC1978,CDC1980,Po1980}. In programming, the intersection type discipline is
of interest because $\lambda$-terms corresponding to correct programs not typeable in the standard Curry type assignment system \cite{curfe},
or in extensions allowing some sort of polymorphism as in ML \cite{Milner78}, are
typeable with intersection types. In \cite{VAK2008} an intersection type system for the $\lambda$-calculus with de Bruijn
indices was introduced, based on the type system given in \cite{KaNo2007}, and proved to satisfy the subject reduction property (SR for short); that is the property of preserving types under $\beta$-reduction: whenever $\Gamma\vdash M\;:\; \sigma$ and $M$ $\beta$-reduces into $N$, $\Gamma\vdash N\;:\; \sigma$.  

A relevant problem in type theory is whether the system has principal
typings (PT for short), which means that for any typeable term $M$
there is a type judgement $\Gamma \vdash M :\tau$ representing all
possible typings $(\Gamma',\tau')$ of $M$ in this system. Expansion
variables are an important process for calculating PT
\cite{CarWe2004}. Since \cite{new}  shows that a typing system
similar to that of \cite{VAK2008} would become incomplete if extended with
expansion variables,
 we did not study the PT property for the system
of \cite{VAK2008}.  Instead, we  consider in this paper a restricted
intersection type system for which we are able to establish the 
PT property for $\beta$-normal forms ($\beta$-nf for short). The concept of a {\it most general} typing is usually linked to syntactic operations and they vary from system to system. For example, the operations to obtain one typing from another in simply typed systems are {\it weakening} and {\it type substitutions}, mapping type variables to types, while in an intersection type system {\it expansion} is performed to obtain intersection types replicating a simple type through some specific rules. In \cite{We2002} J. Wells introduced a system-independent definition of PT and proved that it was the correct generalisation of well known system-dependent definitions such as Hindley's PT for simple type systems \cite{Hi97}. The notion of principal typings has been studied for some intersection type systems (\cite{CDV80}, \cite{roc84}, \cite{roc88}, \cite{bakel95}, \cite{KW2004}) and in \cite{CDV80, roc84} it was proved that PT for some term's $\beta$-nf is principal for the term itself. Partial PT algorithms were proposed in \cite{roc88,KW2004}. In \cite{CarWe2004} S.~Carlier and Wells presented the exact correspondence
between the inference mechanism for their intersection type system and the $\beta$-reduction. They introduce the {\em expansion variables}, integrating expansion operations into the type system (see \cite{CarWeITRS04}). 

We present in this paper a de Bruijn version of the intersection type system originally introduced in \cite{SM96a}, with the purpose of characterising  the syntactic structure of PT for $\beta$-nfs. E.~Sayag and M.~Mauny intended to develop a system where, similarly to simply typed systems, the definition of PT only depends on type substitutions and, as a consequence, their typing system in \cite{SM96a} does not have SR. Although SR is the most basic property and should be satisfied by any typing system, the system infers types to all $\beta$-nfs and, because it is a restriction of more complex and well studied systems, is a reasonable way to characterise PT for intersection type systems. In fact, the system in \cite{SM96a} is a proper restriction of some systems presented in \cite{bakel95}.

Below, we give some definitions and properties for the untyped $\lambda$-calculus with de Bruijn indices, as in \cite{VAK2008}. We introduce the type system in Section \ref{typesystems}, where some properties are stated and counterexamples for some other properties, such as SR, are presented. The type inference algorithm introduced here, its soundness and completeness are at the end of Section \ref{typesystems}. The characterisation of PT for $\beta$-nfs and the reconstruction algorithm are presented  in Section \ref{charactyping}. Both algorithms introduced here are similar to the ones presented in \cite{SM96a}.  

\subsection{$\lambda$-calculus with de Bruijn indices}

\begin{definition}\label{def:estrut}
The set of terms $\Lambda_{dB}$ of \textbf{the $\lambda_{dB}$-calculus}, the $\lambda$-calculus with de Bruijn indices, is defined inductively by:
\begin{tabular}{l}
$M,N\,\in\,\Lambda_{dB} ::=\dbi{n}\;|\,(\app{M}{N})\,|\,\lambda.M\quad\mbox{where} \;n\in \mathbb{N}^*{=}\,\mathbb{N}{\smallsetminus}\{0\}.$
\end{tabular}
\end{definition} 
\begin{definition}\label{def:fi}
$FI(M)$, the {\bf set of free indices} of $M\in\Lambda_{dB}$, is defined by:

\begin{tabular}{c}
 $FI(\dbi{n})\!=\!\{\dbi{n}\} \qquad  FI(\app{M_1}{M_2})\!=\! FI(M_1) \cup FI(M_2) \qquad FI(\lambda.M)\!=\!\{ \dbi{n{-}1},\forall \dbi{n}\in FI(M), n >1\}$
\end{tabular}
\end{definition}
 The free indices correspond to the notion of free variables in the $\lambda$-calculus with names, hence $M$ is called closed when $FI(M)=\emptyset$. The greatest value of $FI(M)$ is denoted by $sup(M)$. In \cite{VAK2008} we give the formal definitions of those concepts. Following, a lemma stating properties about $sup$ related with the structure of terms. 
\begin{lemma}[\cite{VAK2008}]\label{lem:supxstruct}
\begin{enumerate}
\item \label{lem:supxapp}
 $sup(\app{M_1}{M_2})=max(sup(M_1),sup(M_2))$.

\item\label{lem:supxabs} 
If  $sup(M)\!=\!0$, then $sup(\lambda.M)\!=\!0$. Otherwise, $sup(\lambda.M)\!=\!sup(M)-1$. 
\end{enumerate}
\end{lemma}
Terms like $((\dots((M_1\,\,M_2)\,\,M_3)\dots)\,\,M_n)$ are written as
$(M_1\,M_2\,\cdots\,M_n)$, as usual. The $\beta$-contraction definition in this notation needs a mechanism which detects and updates free indices of terms. Intuitively, the {\bf lift} of $M$,  denoted by
$M^+$, corresponds to an increment by $1$ of all free indices occurring in $M$. Thus, we are able to present the definition of the substitution used by $\beta$-contractions, similarly to \cite{ARKa2001a}.
\begin{definition}\label{def:subdb}
Let $m, n\in\mathbb{N}^*$. The \textbf{$\mathbf{\beta}$-substitution} for free occurrences of $\dbi{n}$ in $M\in\Lambda_{dB}$ by term $N$, denoted as $\{\underline{n}\,/N\}M$, is defined inductively by
\vspace{-4mm}
{\small
\begin{displaymath}
 \begin{array}{l@{\hspace{3mm}}l}
  \hspace{-2mm}{\mathit 1.}\,\,\{\underline{n}\,/N\}(M_1\,\,M_2)=(\{\underline{n}\,/N\}M_1\,\,\{\underline{n}\,/N\}M_2) &
   {\mathit 3.}\,\, \{\underline{n}\,/N\}\underline{m}=\left\{%
                          \begin{array}{l@{\hspace{1mm}}l}
                           \underline{m-1}\,, & \textrm{if $m>n$}\\
                           N, & \textrm{if $m=n$}\\
                           \underline{m}\,,& \textrm{if $m<n$}
                          \end{array}%
                          \right.\\[-3mm]
 \hspace{-2mm}{\mathit 2.}\,\,\{\underline{n}\,/N\}(\lambda.M_1) = \lambda.\{\underline{n+1}\,/N^+\}M_1&
\end{array}
\end{displaymath}
}
\end{definition} 
\noindent Observe that in item 2 of Definition \ref{def:subdb}, the lift operator is used to avoid captures of free indices in $N$. We present the $\beta$-contraction as defined in \cite{ARKa2001a}.

\begin{definition}
\textbf{$\mathbf{\beta}$-contraction} in the $\lambda_{dB}$-calculus is defined by $(\lambda.M\,N)\!\to_{\beta}\!\{\underline1\,/N\}M$.
\end{definition}  
\noindent Notice that item 3 in Definition \ref{def:subdb} is the mechanism which does the substitution and updates the free
indices in $M$ as consequence of the lead abstractor
elimination. The {\bf $\mathbf{\beta}$-reduction} is defined to be the $\lambda$-compatible closure of the  $\beta$-contraction defined above. A term is in {\bf $\mathbf{\beta}$-normal form}, $\beta$-nf for short, if there is no possible $\beta$-reduction.
\begin{lemma}
A term $N\in \Lambda_{dB}$ is a $\beta$-nf  iff $N$ is one of the following :
\begin{itemize}
\item[-] $N\equiv\dbi{n}$, for any $n\in\mathbb{N}^*$.

\item[-] $N\equiv \lambda.N'$ and $N'$ is a $\beta$-nf.

\item[-] $N\equiv \dbi{n}\,N_1\cdots N_m$, for some $n\in\mathbb{N}^*$ and $\forall 1\!\leq\!j\!\leq\!m$, $N_j$ is a $\beta$-nf.
\end{itemize} 
\end{lemma}
\begin{proof}
  {\it Necessity} proof is straightforward from $\beta$-nf definition. {\it Sufficiency} proof is by induction on the structure  of $N\!\in\!\Lambda_{dB}$.
\end{proof}

\section{The type system and properties}\label{typesystems}

\begin{definition}\label{def:types}
\begin{enumerate}
\item Let $\mathcal{A}$ be a denumerably infinite {\bf set of type variables} and let $\alpha,\beta$ range over $\mathcal{A}$.

\item The set $\mathcal{T}$ of {\bf restricted intersection types} is defined by:

\begin{tabular}{c}
$\tau,\sigma\,\in\,\mathcal{T} ::= \mathcal{A}\,|\,\ty{\mathcal{U}}{\mathcal{T}}  \hspace{10mm}  
u \,\in\, \mathcal{U} ::= \omega\,|\,\ity{\mathcal{U}}{\mathcal{U}}\,|\,\mathcal{T}$
\end{tabular}

Types are quotiented by taking $\land$ to be commutative,
associative and to have $\omega$ as the neutral element.

\item {\bf Contexts} are ordered lists of $u \in \mathcal{U}$, defined by:
$\Gamma::= nil\,|\,u.\Gamma$

$\Gamma_i$ denotes the $i$-th element of $\Gamma$ and $|\Gamma|$ denotes the length of $\Gamma$. 

$\omega^{\dbi{n}}$ denotes  the sequence $\omega.\omega.\cdots.\omega$ of length $n$ and let $\omega^{\dbi{0}}\,.\Gamma = \Gamma$.

The extension of $\land$ to contexts is done by taking $nil$ as the neutral element and $\ity{(u_1.\Gamma)}{(u_2.\Delta)}=(\ity{u_1}{u_2}).(\ity{\Gamma}{\Delta})$.
Hence, $\land$ is commutative and associative on contexts.

\item {\bf Type substitution} maps type variables to types. Given a type substitution $s\!:\!\mathcal{A} \to \mathcal{T}$,  the corresponding extensions for elements in $\mathcal{U}$ and for contexts are straightforward. The domain of a substitution $s$ is defined by $Dom(s) \!=\! \{ \alpha \,|\, s(\alpha)\neq\alpha\}$ and let $[\alpha/\sigma]$ denote the substitution $s$ such that $Dom(s)\!=\!\{\alpha\}$. For two substitutions $s_1$ and $s_2$ with disjoint domains, let $s_1 + s_2$ be defined by\vspace{-2.5mm}
\[(s_1 + s_2)(\alpha)\left\{ \begin{array}{ll}
                             s_i(\alpha) & \mbox{if}\;\alpha \in Dom(s_i), \mbox{for}\; i\in\{1,2\} \\
                             \alpha      & \mbox{if}\;\alpha \notin Dom(s_1)\cup Dom(s_2)
                             \end{array}\right.
\]

\item $TV(u)$ is the {\bf set of type variables occurring} in $u\in \mathcal{U}$. Extension to contexts is straightforward.

\end{enumerate} 
\end{definition}

The set $\cal{T}$ defined here is equivalent to the one defined in \cite{SM96a}.

\begin{lemma}\label{lem:typeshape}
\begin{enumerate}
\item \label{lem:Ushape}
If $u\!\in\!\mathcal{U}$, then $u\!=\!\omega$ or $u\!=\!\land_{i=1}^{n}\tau_i$
  where $n\!>\!0$ and $\forall\,1\!\leq\! i\!\leq\! n$, $\tau_i\!\in\! \mathcal{T}$.

\item \label{lem:Tshape}
If $\tau\!\in\!\mathcal{T}$, then  $\tau\!=\!\alpha$, $\tau=\ty{\omega}{\sigma}$ or $\tau=\ty{\land_{i=1}^{n}\tau_i}{\sigma}$, where $n>0$ and $\sigma,\tau_1,\dots,\tau_n \in \mathcal{T}$.
\end{enumerate}
\end{lemma}

\begin{proof}
\begin{enumerate}
\item By induction on $u\!\in\!\mathcal{U}$.

\item By induction on $\tau\!\in\!\mathcal{T}$ and Lemma \ref{lem:typeshape}.\ref{lem:Ushape}.\qedhere
\end{enumerate}
\end{proof} 

\begin{definition}\label{def:trules}
\begin{enumerate}

\item The typing rules for system $SM$ are given as follows:

 \begin{table}[h]
 \begin{center}

 {\small
 
 \begin{tabular}[h]{c}

 \begin{tabular}{c@{\hspace{1.3cm}}c}

$\displaystyle\frac{}{\tyj{\dbi{1}}{\tau.nil}{\tau}}\; \mathrm{var}$
& 
$ \displaystyle\frac{\tyj{M}{u.\Gamma}{\tau}}{\tyj{\lambda.M}{\Gamma}{\ty{u}{\tau}}}\;
\to_{i} $\\[3mm]

$\displaystyle\frac{\tyj{\dbi{n}}{\Gamma}{\tau}}{\tyj{\dbi{n{+}1}}{\omega.\Gamma}{\tau}}\;
\mathrm{varn}$ 
&
$\displaystyle\frac{\tyj{M}{nil}{\tau}}{\tyj{\lambda.M}{nil}{\ty{\omega}{\tau}}}\;
\to'_{i}$\\[3mm]
 
\end{tabular} \\

 $\displaystyle\frac{\tyj{M_1}{\Gamma}{\ty{\omega}{\tau}} \qquad
  \tyj{M_2}{\Delta}{\sigma}}{\tyj{\app{M_1}{M_2}}{\ity{\Gamma}{\Delta}}{\tau}}\;\to'_{e}$\\[3mm]

 $\displaystyle\frac{\tyj{M_1}{\Gamma}{\ty{\land_{i=1}^{n}\sigma_i}{\tau}} \qquad
  \tyj{M_2}{\Delta^1}{\sigma_1}\,\dots\,\tyj{M_2}{\Delta^n}{\sigma_n}}{\tyj{\app{M_1}{M_2}}{\ity{\Gamma}{\ity{\Delta^1 \land \cdots}{\Delta^n}}}{\tau}}\;\to_{e}$

\end{tabular}

}
\end{center}
\end{table}\vspace{-4mm}

\item System $SM_r$ is obtained from system $SM$, replacing rule $\mathrm{var}$ by rule
\[ \displaystyle\frac{}{\tyj{\dbi{1}}{\ty{\sigma_1\to\cdots\to\sigma_n}{\alpha}.nil}{\ty{\sigma_1\to\cdots\to\sigma_n}{\alpha}}}\,(n \geq 0)\quad\mathrm{var}_r\]  

\end{enumerate} 

\end{definition}

Type judgements will be of the form $\tyjs{M}{\Gamma}{S}{\tau}$, meaning that term $M$ has type $\tau$ in system $S$ provided $\Gamma$ for $FI(M)$ . Briefly, $M$ has type $\tau$ with $\Gamma$ in $S$ or $(\Gamma,\tau)$ is a typing of $M$ in $S$. The $S$ is omitted whenever its is clear to which system we are referring to.

Note that $SM$ is a proper extension of $SM_r$, hence properties stated for the system $SM$ are also true for the system $SM_r$. The following lemma states that $SM$ is relevant in the sense of \cite{DG94}. 
\begin{lemma}\label{lem:noweak}
If $\tyjs{M}{\Gamma}{SM}{\tau}$, then $|\Gamma|\!=\!sup(M)$ and $\forall 1 \!\leq \!i \!\leq \!|\Gamma|$, $\Gamma_i \neq \omega \:\mbox{iff}\:\dbi{i}\!\in\!FI(M)$.
\end{lemma}
\begin{proof}
By induction on the derivation $\tyj{M}{\Gamma}{u}$.
\begin{itemize}
\item If $\displaystyle\frac{}{\tyj{\dbi{1}}{\tau.nil}{\tau}}$, then $|\Gamma|\!=\! 1 \!=\! sup(\dbi{1}\,)$. Note that $FI(\dbi{1}\,)\!=\!\{\dbi{1}\,\}$ and $\Gamma_1 \!=\! \tau$.

\item If
  $\displaystyle\frac{\tyj{\dbi{n}}{\Gamma}{\tau}}{\tyj{\dbi{n+1}}{\omega.\Gamma}{\tau}}$, then by IH one has $|\Gamma|=sup(\dbi{n}\,)=n$, $\Gamma_n \neq \omega$ and $\forall 1 \leq i < n$, $\Gamma_i = \omega$. Thus, $|\omega.\Gamma| = 1 + |\Gamma| = n\!+\!1 = sup(\dbi{n\!+\!1}\,)$, $(\omega.\Gamma)_{n+1} = \Gamma_{n} \neq \omega$, $(\omega.\Gamma)_1 = \omega$ and $\forall 1\leq i < n$, $(\omega.\Gamma)_{i\!+\!1} = \Gamma_{i} = \omega$.

\item Let $\displaystyle\frac{\tyj{M}{u.\Gamma}{\sigma}}{\tyj{\lambda.M}{\Gamma}{\ty{u}{\sigma}}}$. By IH, $|u.\Gamma| = sup(M)$ and $\forall 0 \leq i \leq sup(M)\!-\!1 $, $(u.\Gamma)_{i\!+\!1} \neq \omega$ iff $\dbi{i\!+\!1}\!\in\!FI(M)$. Hence, $sup(M)= 1\!+\!|\Gamma| > 0$ and, by Lemma \ref{lem:supxstruct}.\ref{lem:supxabs}, $sup(\lambda.M) = sup(M)\!-\!1 = |\Gamma|$. By Definition \ref{def:fi}, $\forall 1\leq i \leq sup(\lambda.M)$, $\dbi{i}\!\in\!FI(\lambda.M)$ iff $\dbi{i\!+\!1}\!\in\!FI(M)$, thus, $(u.\Gamma)_{i\!+\!1} = \Gamma_i \neq \omega$ iff $\dbi{i}\!\in\!FI(\lambda.M)$.

\item Let $\displaystyle\frac{\tyj{M}{nil}{\sigma}}{\tyj{\lambda.M}{nil}{\ty{\omega}{\sigma}}}$.  By IH one has $|nil|\!=\!sup(M)\!=\!0$. Thus, by Lemma \ref{lem:supxstruct}.\ref{lem:supxabs}, $sup(\lambda.M)\!=\!sup(M)\!=\!|nil|$. Note that $FI(M)\!=\!FI(\lambda.M)\!=\!\emptyset$.

\item Let $\displaystyle\frac{\tyj{M_1}{\Gamma}{\ty{\omega}{\tau}} \qquad
  \tyj{M_2}{\Delta}{\sigma}}{\tyj{\app{M_1}{M_2}}{\ity{\Gamma}{\Delta}}{\tau}}$.  By IH, $|\Gamma| {=} sup(M_1)$, $\forall 1\!\leq\! i \!\leq\! |\Gamma|$ one has $\Gamma_i \neq \omega\,\mbox{iff}\, \dbi{i}\!\in\!FI(M_1)$, $|\Delta| = sup(M_2)$ and $\forall 1\leq j\leq |\Delta|$ one has $\Delta_j \neq \omega \,\mbox{iff}\, \dbi{j}\!\in\!FI(M_2)$. By Lemma \ref{lem:supxstruct}.\ref{lem:supxapp} one has $sup(\app{M_1}{M_2}) = max(sup(M_1),sup(M_2))\!= \!max(|\Gamma|,|\Delta|) = |\ity{\Gamma}{\Delta}|$. Let $1\!\leq\!l\!\leq\!|\ity{\Gamma}{\Delta}|$ and suppose w.l.o.g. that $l \leq |\Gamma|,|\Delta|$. Thus, $(\ity{\Gamma}{\Delta})_l = \ity{\Gamma_l}{\Delta_l} \!\neq\! \omega \,\mbox{\it iff}\, \Gamma_l \!\neq\! \omega$ or $\Delta_l \!\neq\! \omega  \,\mbox{\it iff}\, \dbi{l} \!\in\! FI(M_1)$ or $\dbi{l} \!\in\! FI(M_2) \,\mbox{\it iff}\, \dbi{l} \!\in\! FI(M_1)\cup FI(M_2) \!=\! FI(\app{M_1}{M_2})$.

\item Let $\displaystyle\frac{\tyj{M_1}{\Gamma}{\ty{\land_{k=1}^{n}\sigma_k}{\tau}} \qquad
  \tyj{M_2}{\Delta^1}{\sigma_1}\,\dots\,\tyj{M_2}{\Delta^n}{\sigma_n}}{\tyj{\app{M_1}{M_2}}{\ity{\Gamma}{\ity{\Delta^1\!\land \cdots}{\Delta^n}}}{\tau}}$. By IH, $|\Gamma| = sup(M_1)$, $\forall 1\!\leq \!i\!\leq\! |\Gamma|$ one has $\Gamma_i \neq \omega\,\mbox{iff}\, \dbi{i}\!\in\!FI(M_1)$ and $\forall 1 \!\leq \!k\! \leq\! n$, $|\Delta^k| = sup(M_2)$ and $\forall 1\!\leq\! j\!\leq\! |\Delta^k|$ one has $\Delta_j^k \neq \omega \,\mbox{iff}\, \dbi{j}\!\in\!FI(M_2)$. Let $\Delta' = \ity{\Delta^1\!\land \cdots\,}{\Delta^n}$. Thus, $|\Delta'| = sup(M_2)$ and $\forall 1\!\leq\! j \!\leq \! |\Delta'|$, $\Delta'_j \neq \omega \,\mbox{\it iff}\, \dbi{j} \in FI(M_2)$. The proof is analogous to the one above.\qedhere
\end{itemize}
\end{proof}
Note that, by Lemma \ref{lem:noweak} above, system $SM$ is not only relevant but there is a strict relation between the free indices of terms and the length of contexts in their typings. Following, a generation lemma is presented for typings in $SM$ and some specific for $SM_r$
\begin{lemma}[Generation]\label{lem:gen}
\begin{enumerate}
\item \label{lem:genvar}
If $\tyjs{\dbi{n}}{\Gamma}{SM}{\tau}$, then $\Gamma_n\! =\!\tau$.

\item \label{lem:genvarr}
If $\tyjs{\dbi{n}}{\Gamma}{SM_r}{\tau}$, then $\tau = \ty{\sigma_{1}\to\cdots\to \sigma_{k}}{\alpha}$ for $k\geq 0$.

\item \label{lem:genabsone}
If $\tyjs{\lambda.M}{nil}{SM}{\tau}$, then either $\tau\!=\!\ty{\omega}{\sigma}$  and $\tyj{M}{nil}{\sigma}$ or $\tau\!=\!\ty{\land_{i=1}^{n}\sigma_i}{\sigma}$, $n>0$, and $\tyjs{M}{\land_{i=1}^{n}\sigma_i.nil}{SM}{\sigma}$ for some $\sigma,\sigma_1,\dots,\sigma_n\!\in\!\mathcal{T}$.

\item \label{lem:genabstwo} 
If $\tyjs{\lambda.M}{\Gamma}{SM}{\tau}$ and $|\Gamma|>0$, then $\tau\!=\!\ty{u}{\sigma}$ for some $u\!\in\!\mathcal{U}$ and $\sigma\!\in\!\mathcal{T}$, where $\tyjs{M}{u.\Gamma}{SM}{\sigma}$.

\item \label{lem:genapp}
If $\tyjs{\dbi{n}\;M_1\cdots M_m}{\Gamma}{SM_r}{\tau}$, then $\Gamma = \ity{(\omega^{\dbi{n\!-\!1}}\,.\ty{\sigma_1\to\cdots\to \sigma_{m}}{\tau}.nil)}{\Gamma^1\!\land\cdots\land \Gamma^m}$, $\forall 1\!\leq\!i\!\leq\!m$, $\tyjs{M_i}{\Gamma^i}{SM_r}{\sigma_i}$ and $\tau = \ty{\sigma_{m\!+\!1}\to\cdots\to \sigma_{m\!+\!k}}{\alpha}$.

\end{enumerate}
\end{lemma}
\begin{proof}
\begin{enumerate}
\item By induction on the derivation $\tyjs{\dbi{n}}{\Gamma}{SM}{\tau}$. Note that $(\omega.\Gamma)_{n\!+\!1} = \Gamma_{n}$.

\item By induction on the derivation $\tyjs{\dbi{n}}{\Gamma}{SM_r}{\tau}$.

\item By case analysis on the derivation $\tyjs{\lambda.M}{nil}{SM}{\tau}$.

\item By case analysis on the derivation $\tyjs{\lambda.M}{\Gamma}{SM}{\tau}$, for $|\Gamma|>0$.

\item By induction on $m$.

If $m=0$, then, by Lemma \ref{lem:gen}.\ref{lem:genvarr}, $\tau = \ty{\sigma_{1}\to\cdots\to \sigma_{k}}{\alpha}$. Thus, by Lemmas \ref{lem:noweak} and \ref{lem:gen}.\ref{lem:genvar}, $\Gamma = \omega^{\dbi{n\!-\!1}}\,.\tau.nil$.

If $m=m'+1$, then by case analysis the last step of the derivation is
\[ \displaystyle\frac{\tyj{\dbi{n}\;M_1\cdots M_{m'}}{\Gamma}{\ty{\land_{j=1}^{l}\tau_j}{\tau}} \quad
  \tyj{M_{m'\!+\!1}}{\Delta^1}{\tau_1}\,\dots\,\tyj{M_{m'\!+\!1}}{\Delta^l}{\tau_l}}{\tyj{\app{\dbi{n}\;M_1\cdots M_{m'}}{M_{m'\!+\!1}}}{\ity{\Gamma}{\ity{\Delta^1\!\land \cdots}{\Delta^l}}}{\tau}}\]
By IH, $\Gamma \!=\! \ity{(\omega^{\dbi{n\!-\!1}}\,.\ty{\sigma_1\to\cdots\to \sigma_{m'}}{(\ty{\land_{j=1}^{l}\tau_j}{\tau})}.nil)}{\Gamma^1\!\land\cdots\land \Gamma^{m'}}$, $\forall 1\!\leq\!i\!\leq\!m'$, $\tyjs{M_i}{\Gamma^i}{SM_r}{\sigma_i}$ and $\ty{\land_{j=1}^{l}\tau_j}{\tau} \!=\! \ty{\sigma_{m'\!+\!1}\to\cdots\to \sigma_{m'\!+\!k}}{\alpha}$ . Therefore, $\tau = \ty{\sigma_{m'\!+\!2}\to\cdots\to \sigma_{m\!+\!k}}{\alpha}$, $l\!=\!1$ and $\tau_1\!=\! \sigma_{m'\!+\!1}$. Hence, taking $\Gamma^{m'\!+\!1} = \Delta^1$ and $\sigma_{m'\!+\!1} = \tau_1$, the result holds.\qedhere 

\end{enumerate}
\end{proof}
Following, we will give counterexamples to show that neither subject expansion nor reduction holds. 
\begin{example}
In order to have the subject expansion property, we need to prove the statement: If $\tyj{\{\dbi{1}\,/N\}M}{\Gamma}{\tau}$ then $\tyj{(\app{\lambda.M}{N})}{\Gamma}{\tau}$. Let $M\equiv \lambda.\dbi{1}$ and $N\equiv \dbi{3}$, hence $\{\dbi{1}\,/\dbi{3}\,\}\lambda.\dbi{1} = \lambda.\dbi{1}$. We have that, by generation lemmas, $\tyj{\lambda.\dbi{1}}{nil}{\ty{\alpha}{\alpha}}$. Thus, $\tyj{\lambda.\lambda.\dbi{1}}{nil}{\ty{\omega}{\ty{\alpha}{\alpha}}}$ and $\tyj{\dbi{3}}{\omega.\omega.\beta.nil}{\beta}$, then $\tyj{\app{\lambda.\lambda.\dbi{1}}{\dbi{3}}}{\omega.\omega.\beta.nil}{\ty{\alpha}{\alpha}}$.
\end{example}
For subject reduction, we need the statement: If $\tyj{(\app{\lambda.M}{N})}{\Gamma}{\tau}$ then $\tyj{\{\dbi{1}\,/N\}M}{\Gamma}{\tau}$. Note that if we take $M$ and $N$ as in the example above, we have the same problem as before but in the other way round. In other words, we have a restriction on the original context after the $\beta$-reduction, since we loose the typing information regarding $N \equiv \dbi{3}$. 

One possible solution for those problems is to replace rule $\to_e'$ by
\begin{tabular}{c}
$\displaystyle\frac{\tyj{M}{\Gamma}{\ty{\omega}{\tau}}}{\tyj{\app{M}{N}}{\Gamma}{\tau}}$
\end{tabular}
  
This approach was originally presented in \cite{SM96b}, but a new notion replacing free index should be introduced since we would not have the typing information for all free indices occurring in a term. In \cite{SM96b}, and in \cite{SM97}, no notion is presented instead of the usual free variables, which is wrongly used to state things that are not actually true.  

The other way to achieve the desired properties is to think about the meaning of the properties itself. Since, by  Lemma \ref{lem:noweak}, the system is related to relevant logic (see \cite{DG94}), the notion of restriction of contexts is an interesting way to talk about subject reduction. This concept was presented in \cite{KaNo2007} for environments, where environments expansion was also introduced for the sake of subject expansion. Note that this approach is not sufficient to regain subject expansion for system $SM$, since in rule $\to'_e$ it is required that the term being applied is also typeable.

Even though, any $\beta$-nf is typeable with system $SM_r$. We introduce the type inference algorithm $\mathtt{Infer}$ for $\beta$-nfs, similarly to \cite{SM96a}. 
\begin{definition}[Type inference algorithm]

Let $N$ be a $\beta$-nf:
\newline

{\small
$\mathtt{Infer}(N) = $\\

     \begin{tabular}{cl}
      {\bf Case} & $N = \dbi{n}$ \\

                 & {\bf let} $\alpha$ be a fresh type variable \\

                 & {\bf return} $(\omega^{\dbi{n{-}1}}\,.\alpha.nil,\alpha)$       
     \end{tabular}

     \begin{tabular}{cl}
     {\bf Case} & $N = \lambda.N'$ \\

                & {\bf let} $(\Gamma',\sigma) = \mathtt{Infer}(N')$ \\ 

                & {\bf if} $(\Gamma' = u.\Gamma)$ {\bf then} \\ 

                & {\bf return} $(\Gamma,\ty{u}{\sigma})$ \\

                & {\bf else} \\

                & {\bf  return} $(nil,\ty{\omega}{\sigma})$
     \end{tabular}

     \begin{tabular}{cl}
     {\bf Case} & $N = \app{\dbi{n}\,N_1\cdots}{N_m}$ \\

                & {\bf let} $(\Gamma^1,\sigma_1) = \mathtt{Infer}(N_1)$ \\
                &           \hspace{1.7cm} $\vdots$ \\
                &           \hspace{.2cm} $(\Gamma^m,\sigma_m) = \mathtt{Infer}(N_m)$\\
                &           \hspace{.2cm} $\alpha$ be a fresh type variable\\
                & {\bf return} $(\ity{(\omega^{\dbi{n{-}1}}\,.\ty{\sigma_1\to\cdots\to\sigma_m}{\alpha}.nil)}{\Gamma^1\!\land \cdots \land \Gamma^m},\alpha)$
     \end{tabular}

}

\end{definition}

\noindent Similarly to \cite{SM96a}, the notion of {\it fresh type variables} is used to prove completeness. The freshness of a variable is to guarantee that each time some type variable is picked up from $\mathcal{A}$ it is a new one. Therefore, two non overlapped calls to $\mathtt{Infer}$  return pairs with disjoints sets of type variables. Below, a runnig example of how the algorithm is applied is presented.

\begin{example}
Let $N \equiv \dbi{2}\:\:(\lambda.\!\dbi{1}\,)\:\:\dbi{1}\:\:\,\lambda.\app{\dbi{1}}{\dbi{1}\,}$. For $\mathtt{Infer}(N)$, the term $N$ matches the third case, for $n=2$. The algorithm is then called recursively as follows
\begin{eqnarray*}
 (\Gamma^1,\sigma_1)&=& \mathtt{Infer}(\lambda.\!\dbi{1}\,)\\
 (\Gamma^2,\sigma_2)&=& \mathtt{Infer}(\dbi{1}\,)\\
 (\Gamma^3,\sigma_3)&=&  \mathtt{Infer}(\lambda.\!\app{\dbi{1}}{\dbi{1}\,})
\end{eqnarray*}
Below, we show how each call is treated by the algorithm.

The case $\mathtt{Infer}(\lambda.\!\dbi{1}\,)$ goes down recursively to obtain $\mathtt{Infer}(\dbi{1}\,)=(\alpha_1.nil,\alpha_1)$ and then one has that $\mathtt{Infer}(\lambda.\!\dbi{1}\,)= (nil,\ty{\alpha_1}{\alpha_1})$. 

The case $\mathtt{Infer}(\dbi{1}\,)$ returns
$(\alpha_2.nil,\alpha_2)$. Note that we have to take a different type variable
from the one used in the previous case.

The case $\mathtt{Infer}(\lambda.\!\app{\dbi{1}}{\dbi{1}\,})$ goes down recursively to return $\mathtt{Infer}(\dbi{1}\,)=(\alpha_3.nil,\alpha_3)$, for the subterm $\dbi{1}\,$ on the right. For a fresh type variable $\alpha_4$, one has that $\ity{\ty{\alpha_3}{\alpha_4}.nil}{\alpha_3.nil} = \ity{(\ty{\alpha_3}{\alpha_4})}{\alpha_3}.nil$. Hence, $\mathtt{Infer}\app{\dbi{1}}{\dbi{1}\,} = (\ity{(\ty{\alpha_3}{\alpha_4})}{\alpha_3}.nil,\alpha_4)$. Finally, $\mathtt{Infer}(\lambda.\!\app{\dbi{1}}{\dbi{1}\,}) = (nil, \ity{(\ty{\alpha_3}{\alpha_4})}{\alpha_3} \to \alpha_4 )$.

Now, let $\tau = \ty{(\ty{\alpha_1}{\alpha_1}) \to \alpha_2 \to (\ity{(\ty{\alpha_3}{\alpha_4})}{\alpha_3} \to \alpha_4)}{\alpha_5}$ for the fresh type variable $\alpha_5$. One has that $(\omega.\tau)\land nil \land (\alpha_2.nil) \land nil = \alpha_2.\tau.nil$. Therefore, $\mathtt{Infer}(N) = (\alpha_2.\tau.nil,\alpha_5)$.
\end{example}
 
\begin{theorem}[Soundness]\label{teo:sound}
If $N$ is a $\beta$-nf and $\mathtt{Infer}(N)=(\Gamma,\sigma)$, then $\tyjs{N}{\Gamma}{SM_r}{\sigma}$.
\end{theorem} 
\begin{proof}
By structural induction on $N$.

\begin{itemize}
\item If $N \equiv \dbi{n}$ then $\mathtt{Infer}(\dbi{n}\,) = (\omega^{\dbi{n{-}1}}\,.\alpha.nil, \alpha)$. By rule $\mathrm{var}_r$, $\tyj{\dbi{1}}{\alpha.nil}{\alpha}$ and, by rule $\mathrm{varn}$ applied $n{-}1$ times, $\tyj{\dbi{n}}{\omega^{\dbi{n\!-\!1}}\,.\alpha.nil}{\alpha}$. 

\item Let $N \equiv \lambda.N'$. If $(\Gamma',\sigma) \!=\! \mathtt{Infer}(N')$ then, by IH one has $\tyj{N'}{\Gamma'}{\sigma}$. Thus, if $\Gamma'\!=\! u.\Gamma$ then $ \mathtt{Infer}(\lambda.N') \!=\! (\Gamma,\ty{u}{\sigma})$ and, by rule $\to_{i}$, $\tyj{\lambda.N'}{\Gamma}{\ty{u}{\sigma}}$, otherwise one has $ \mathtt{Infer}(\lambda.N') \!=\! (nil,\ty{\omega}{\sigma})$ and, by rule $\to'_{i}$,  $\tyj{\lambda.N'}{nil}{\ty{\omega}{\sigma}}$.

\item Let $N \equiv \dbi{n}\,N_1\cdots N_m$. If $\forall 1{\leq}i{\leq}m$, $(\Gamma^i,\sigma_{i}) \!=\! \mathtt{Infer}(N_i)$ then, by IH, $\forall 1{\leq}i{\leq}m$, $\tyj{N_i}{\Gamma^i}{\sigma_i}$. Let $\Delta = \omega^{\dbi{n{-}1}}\,.\ty{\sigma_1\to\cdots\to\sigma_m}{\alpha}.nil$. Hence $\mathtt{Infer}(N)=(\ity{\Delta}{\Gamma^1\!\land \cdots \land \Gamma^m},\alpha)$ for some fresh type variable $\alpha$. By rule $\mathrm{var}_r$ and  by  rule $\mathrm{varn}$ $n{-}1$-times, $\tyj{\dbi{n}}{\Delta}{\ty{\sigma_1\to\cdots\to\sigma_m}{\alpha}}$ and, by rule $\to_{e}$ $m$-times, $\tyj{N}{\ity{\Delta}{\Gamma^1\!\land \cdots \land \Gamma^m}}{\alpha}$.\qedhere

\end{itemize}
\end{proof}
Note that, since the choice of the new type variables is not fixed, $\mathtt{Infer}$ is well defined up to the name of type variables.
\begin{coro}\label{coro:nftyping}
If $N$ is a $\beta$-nf then $N$ is typeable in system $SM_r$.
\end{coro} 

\begin{theorem}[Completeness]\label{teo:complete}
If $\tyjs{N}{\Gamma}{SM_r}{\sigma}$, $N$ a $\beta$-nf, then for $(\Gamma',\sigma')=\mathtt{Infer}(N)$ exists a type substitution $s$ such that $s(\Gamma')=\Gamma$ and $s(\sigma') = \sigma$.
\end{theorem}

\begin{proof}
By structural induction on $N$
\begin{itemize}
\item Let $N \equiv \dbi{n}$. If $\tyj{\dbi{n}}{\Gamma}{\sigma}$ then, by Lemmas \ref{lem:noweak} and \ref{lem:gen}.\ref{lem:genvar}, $\Gamma \!=\! \omega^{\dbi{n{-}1}}\,.\sigma.nil$. One has that $\mathtt{Infer}(\dbi{n}\,)= (\omega^{\dbi{n{-}1}}\,.\alpha.nil,\alpha)$, then take $s=[\alpha/\sigma]$.

\item Let $N \equiv \lambda.N'$ and suppose that $\tyj{\lambda.N'}{\Gamma}{\sigma}$.

If $\Gamma \!=\! nil$, then by Lemma \ref{lem:gen}.\ref{lem:genabsone} either $\sigma \!=\! \ty{\omega}{\tau}$ and $\tyj{N'}{nil}{\tau}$ or $\sigma \!=\! \ty{\land_{j=1}^{n}\sigma_j}{\tau}$ and $\tyj{N'}{\land_{j=1}^{n}\sigma_j.nil}{\tau}$. The former, by IH, $\mathtt{Infer}(N') \!=\! (\Gamma',\tau')$  and there exists $s$ s.t. $s(\tau') \!=\! \tau$ and $s(\Gamma') \!=\! nil$, thus $\Gamma'\!=\!nil$. Hence, $\mathtt{Infer}(\lambda.N') \!=\! (nil,\ty{\omega}{\tau'})$ and $s(\ty{\omega}{\tau'}) \!=\! \ty{s(\omega)}{s(\tau')} \!=\! \sigma$. The latter, by IH, $\mathtt{Infer}(N')\!=\!(\Gamma',\tau')$ and there exists $s$ s.t. $s(\tau') \!=\! \tau$ and $s(\Gamma' ) \!=\! \land_{j=1}^{n}\sigma_j.nil$. Then $\Gamma'\!=\!u.nil$ for $s(u)\!=\!\land_{j=1}^{n}\sigma_j$, hence $\mathtt{Infer}(\lambda.N') \!=\! (nil,\ty{u}{\tau'})$ and $s(\ty{u}{\tau'}) \!=\! \ty{s(u)}{s(\tau')} \!=\!\sigma$.

Otherwise, by Lemma \ref{lem:gen}.\ref{lem:genabstwo}, $\sigma\!=\! \ty{u}{\tau}$ and $\tyj{N'}{u.\Gamma}{\tau}$. The proof is analogous to the one above. 

\item Let $N \equiv \app{\dbi{n}\,N_1\cdots}{N_m}$. If $\tyj{\dbi{n}\,N_1\cdots N_m}{\Gamma}{\sigma}$ then, by Lemma \ref{lem:gen}.\ref{lem:genapp}, $\forall 1\!\leq\!i\!\leq\!m$, $\tyj{N_i}{\Gamma^i}{\sigma_i}$ s.t. $\Gamma \!=\! \ity{(\omega^{\dbi{n\!-\!1}}\,.\ty{\sigma_1\to\cdots\to \sigma_{m}}{\sigma}.nil)}{\Gamma^1\!\land\cdots\land \Gamma^m}$. By IH,  $\forall 1\!\leq\!i\!\leq\!m$, $\mathtt{Infer}(N_i) \!=\! (\Gamma^{i'},\sigma'_i)$ and there is a $s_i$ s.t. $s_i(\sigma'_i) \!=\! \sigma_i$ and $s_i(\Gamma^{i'})\!=\!\Gamma^{i}$. One has that $\mathtt{Infer}(N) \!=\! (\ity{(\omega^{\dbi{n{-}1}}\,.\ty{\sigma'_1\to\cdots\to\sigma'_m}{\alpha}.nil)}{\Gamma^{1'}\!\land \cdots \land \Gamma^{m'}},\alpha)$, for some fresh type variable $\alpha$. The domain of each $s_i$ is compounded by the type variables returned by each call of $\mathtt{Infer}$ for the corresponding $N_i$, consequently they are disjoint. Thus, for $s \!=\! [\alpha/\sigma] + s_1 + \cdots + s_m$ the result holds. \qedhere

\end{itemize} 
\end{proof}
Hence, the pair returned by $\mathtt{Infer}$ for some $\beta$-nf $N$ is a most general typing of $N$ is $SM_r$. Note that these typings are unique up to renaming of type variables.
\begin{coro}
If $N$ is a $\beta$-nf, then $(\Gamma,\sigma)=\mathtt{Infer}(N)$ is a principal typing of $N$ in $SM_{r}$.
\end{coro}

\section{Characterisation of principal typings}\label{charactyping}

Following, we give some characterisation of principal typings for $\beta$-nfs, analogue to \cite{SM96a}. To begin with, we introduce proper subsets of $\mathcal{T}$ and $\mathcal{U}$ containing the pairs returned by $\mathtt{Infer}$.         
\begin{definition}
\begin{enumerate}
\item Let $\mathcal{T}_C$, $\mathcal{T}_{NF}$ and $\mathcal{U}_C$ be defined by:

\begin{tabular}{l}
$\rho \,\in\, \mathcal{T}_C ::= \mathcal{A}\,|\,\ty{\mathcal{T}_{NF}}{\mathcal{T}_C} \hspace{.35in}
\varphi \,\in\,\mathcal{T}_{NF} ::= \mathcal{A}\,|\,\ty{\mathcal{U}_{C}}{\mathcal{T}_{NF}}\hspace{.35in}
v \,\in\,\mathcal{U}_C ::= \omega\,|\,\ity{\mathcal{U}_{C}}{\mathcal{U}_{C}}\,|\,\mathcal{T}_C$
\end{tabular}

\item Let $\mathcal{C}$ be the set of contexts $\Gamma::= nil \,|\, v.\Gamma$ such that $v \in \mathcal{U}_C$. Observe that $\mathcal{C}$ is closed under $\land$.

\end{enumerate}
\end{definition}
\begin{lemma}\label{lem:infersubset}
If $\mathtt{Infer}(N)=(\Gamma,\sigma)$, $N$ a $\beta$-nf, then $(\Gamma,\sigma) \in \mathcal{C}\!\times\!\mathcal{T}_{NF}$.
\end{lemma}
\begin{proof}
By structural induction on $N$.
\end{proof}

\begin{definition}
Let $Im(\mathtt{Infer})$ be defined as the set of pairs $(\Gamma,\sigma)=\mathtt{Infer}(N)$ for some $\beta$-nf $N$.
\end{definition}
\begin{coro}
$Im(\mathtt{Infer}) \subseteq  \mathcal{C}\!\times\!\mathcal{T}_{NF}$.
\end{coro} 
We use the usual notion of {\bf positive} and {\bf negative} occurrences of type variables  and of {\bf final occurrences} for elements $u\in\mathcal{U}$ (see \cite{Kri93}). For contexts, the positive and negative occurrences are the respective occurrences in the types forming the contexts' sequences. 
\begin{definition}\label{def:Gtypes}
Let $\Gamma \in \mathcal{C}$ and $\varphi\in\mathcal{T}_{NF}$. The {\bf $\mathbf{\mathcal{C}}$-types} $T$ are defined by: 
\begin{tabular}{c}
$ T ::= \cty{\Gamma}{\varphi}  \,|\, \cty{\Delta}{}  \quad \mbox{s.t.}\;|\Delta| > 0$
\end{tabular}
\end{definition}
Note that, for any $\beta$-nf $N$, $\mathtt{Infer}(N)$ has a unique corresponding $\mathcal{C}$-type $T^{N}$. The corresponding $A$-types in \cite{SM96a} are defined by taking the set of multisets associated to an environment and transforming them in a single multiset used on the left hand of $\Rightarrow$. Thus, for an environment $A$ and type $\tau$, $\cty{\overline{A}}{\tau}$ is the $A$-type with $\overline{A}$ being the multiset obtained from $A$. On Definition \ref{def:Gtypes} above the sequential structure of contexts are preserved.
\begin{definition}\label{def:held}
Let $T=\cty{\Gamma}{\varphi}$ be a $\mathcal{C}$-type, $T'$ is {\bf held} in $T$ if $T'=\cty{\Gamma'}{}$ or $\cty{\Gamma'}{\varphi}$, such that $\Gamma = \ity{\Gamma'}{\Delta}$ for $\Gamma'\neq \omega^{\dbi{n}}$ and some context $\Delta$. If $T'\neq T$ then $T'$ is {\bf strictly held} in $T$.
\end{definition}
Observe that on Definition \ref{def:held} above  we have that $\Gamma'$ can be $nil$ for $T' = \cty{\Gamma'}{\varphi}$ and $\Delta = \omega^{\dbi{n}}$ for any $n\leq |\Gamma|$ when $\Gamma'=\Gamma$.

\begin{definition}\label{def:leftst}
The set $L(T)$ of the {\bf left subtypes} for some $\mathcal{C}$-type $T$ is defined by structural induction:
\begin{itemize}
\item[-] $L(\cty{\Gamma}{})=L(\Gamma)$.

\item[-] $L(\cty{\Gamma}{\varphi})=L(\Gamma)\cup L(\varphi)$.
 
\item[-]$L(v.\Gamma) = \{v\}\cup L(\Gamma)$ if $v\neq\omega$ and $ L(\Gamma)$ otherwise.

\item[-]$L(nil) = \emptyset$.

\item[-] $L(\ty{v}{\varphi}) = \{v\}\cup L(\varphi)$ if $v\neq\omega$ and $L(\varphi)$ otherwise.

\item[-] $L(\alpha) = \emptyset$.

\end{itemize}
\end{definition}

The notion of sign of occurrences for type variable are straightforward extended to $\mathcal{C}$-types, where the polarity changes on the left side of $\Rightarrow$. We have that $TV(\cty{\Gamma}{\varphi}) = TV(\Gamma)\cup TV(\varphi)$.

\begin{definition}\label{def:closedGt}
A $\mathcal{C}$-type $T$ is {\bf closed} if each $\alpha\in TV(T)$ has exactly one positive and one negative occurrences in $T$.
\end{definition}

\begin{lemma}\label{lem:closedprop}
\begin{enumerate}
\item\label{closed1}
$\cty{v.\Gamma}{\varphi}$ is closed iff $\cty{\Gamma}{\ty{v}{\varphi}}$ is closed.

\item\label{closed2}
$\cty{nil}{\varphi}$ is closed iff $\cty{nil}{\ty{\omega}{\varphi}}$ is closed.

\item\label{closed3}
If $\forall 1 \!\leq\! i\!\leq\! m$, $T_i = \cty{\Gamma^i}{\varphi_i}$ is closed and $TV(T_i)$ are pairwise disjoint then, for any fresh type variable $\alpha$,  $\cty{\ity{(\omega^{\dbi{n-1}}.\ty{\varphi_1\to\cdots\to\varphi_m}{\alpha}.nil)}{\Gamma^1\!\land\cdots\land\Gamma^m}}{\alpha}$ is closed.

\end{enumerate}
\end{lemma}

\begin{proof}
\begin{enumerate}
\item  Let $T=\cty{v.\Gamma}{\varphi}$ and $T'=\cty{\Gamma}{\ty{v}{\varphi}}$. Note that $TV(T) = TV(T')$ and that the sign for type variable occurrences in $v$ for both $T$ and $T'$ are exactly the same.

\item analogous to the proof above.

\item Let $T=\cty{\ity{(\omega^{\dbi{n-1}}.\ty{\varphi_1\to\cdots\to\varphi_m}{\alpha}.nil)}{\Gamma^1\!\land\cdots\land\Gamma^m}}{\alpha}$. Since $TV(T_i)$ are pairwise disjoint, $TV(T) \!=\! \cup_{i=1}^{m}TV(T_i) \cup \{\alpha\}$ and $T$ has exactly two occurrences of each type variable. Note that $\forall 1{\leq}i{\leq}m$ the type variable occurrences in $\Gamma^i$ and $\varphi_i$ have exactly the same sign on both $T_i$ and $T$ and that $\alpha$ has one positive and one negative occurrence in $T$. Hence, $T$ is closed.\qedhere 
\end{enumerate}
\end{proof}

\begin{definition}\label{def:fclosedGt}
A $\mathcal{C}$-type $T = \cty{\Gamma}{\varphi}$ is {\bf finally closed}, f.c. for short, if the final occurrence of $\varphi$ is also the final occurrence of a type in $L(T)$.
\end{definition}

\begin{lemma}\label{lem:fclosedprop}
\begin{enumerate}
\item\label{fc1}
$\cty{v.\Gamma}{\varphi}$ is finally closed iff $\cty{\Gamma}{\ty{v}{\varphi}}$ is finally closed.

\item\label{fc2}
$\cty{nil}{\varphi}$ is finally closed iff $\cty{nil}{\ty{\omega}{\varphi}}$ is finally closed.

\end{enumerate}
\end{lemma}

\begin{proof}
\begin{enumerate}
\item  Let $T \!=\! \cty{v.\Gamma}{\varphi}$ and $T' \!=\! \cty{\Gamma}{\ty{v}{\varphi}}$. The final occurrence of $\ty{v}{\varphi}$ is the same as of $\varphi$. If $v\!\neq\!\omega$, by Definition \ref{def:leftst}, $L(T) {=} L(v.\Gamma) \!\cup\! L(\varphi){=} \{v\} \!\cup\! L(\Gamma) \!\cup\! L(\varphi) {=} L(\Gamma) \!\cup\! L(\ty{v}{\varphi}) {=} L(T')$. Otherwise, $L(T){=} L(\omega.\Gamma) \!\cup\! L(\varphi){=} L(\Gamma) \!\cup\! L(\varphi) {=} L(\Gamma) \!\cup\! L(\ty{\omega}{\varphi}) {=} L(T')$. Hence, $T$ is f.c. iff $T'$ is f.c.

\item analogous to the proof above.\qedhere

\end{enumerate}
\end{proof}

\begin{definition}\label{def:mclosedGt}
A $\mathcal{C}$-type $T$ is {\bf minimally closed}, m.c. for short, if there is no closed $T'$ strictly held in $T$.
\end{definition}
\begin{lemma}\label{lem:mclosedprop}
\begin{enumerate}
\item\label{mc1}
If $\cty{v.\Gamma}{\varphi}$ is m.c. for $v\neq\omega$, then $\cty{\Gamma}{\ty{v}{\varphi}}$ is m.c.

\item\label{mc2}
$\cty{\omega.\Gamma}{\varphi}$ is m.c. iff $\cty{\Gamma}{\ty{\omega}{\varphi}}$ is m.c.

\item\label{mc3}
$\cty{nil}{\varphi}$ is m.c. iff $\cty{nil}{\ty{\omega}{\varphi}}$ is m.c.

\item\label{mc4}
If $\forall 1 \!\leq\! i \!\leq\! m$, $T_i = \cty{\Gamma^i}{\varphi_i}$ is m.c. and $TV(T_i)$ are pairwise disjoint then, for any fresh type variable $\alpha$, $T = \cty{\ity{(\omega^{\dbi{n-1}}.\ty{\varphi_1\to\cdots\to\varphi_m}{\alpha}.nil)}{\Gamma^1\!\land\cdots\land\Gamma^m}}{\alpha}$ is m.c..
\end{enumerate}
\end{lemma}

\begin{proof}
\begin{enumerate}
\item Let $T = \cty{v.\Gamma}{\varphi}$ be m.c. for $v\!\neq\!\omega$ and let $T'=\cty{\Gamma}{\ty{v}{\varphi}}$. Let $T''$ be strictly held in $T'$. If $T'' = \cty{\Gamma'}{\ty{v}{\varphi}}$ then $T''' = \cty{v.\Gamma'}{\varphi}$ is strictly held in $T$. By Lemma \ref{lem:closedprop}.\ref{closed1}, $T''$ is closed iff $T'''$ is closed. Thus, since $T$ is m.c., $T''$ cannot be closed. If $T'' \!=\! \cty{\Gamma'}{}$ then one has similarly that $T''$ cannot be closed. Hence, $T'$ is m.c..

\item Let $T$ be strictly held in $\cty{\omega.\Gamma}{\varphi}$. One has that $T = \cty{\omega.\Gamma'}{\varphi}$ is strictly held in $\cty{\omega.\Gamma}{\varphi}$ iff $T'=\cty{\Gamma'}{\ty{\omega}{\varphi}}$ is strictly held in $\cty{\Gamma}{\ty{\omega}{\varphi}}$. There is a corresponding $T'$ for $T = \cty{nil}{\varphi}$ and for $T=\cty{\omega.\Gamma'}{}$. Therefore, by Lemma \ref{lem:closedprop}.\ref{closed1}, there is a closed $T$ strictly held in $\cty{\omega.\Gamma}{\varphi}$ iff there is a closed $T'$ strictly held in $\cty{\Gamma}{\ty{\omega}{\varphi}}$.

\item analogous to the proof above.

\item Let $T'$ be held in $T$ defined above and suppose that $T'$ is closed. If $T' = \cty{\Gamma'}{}$ then, since $|\Gamma'|\!>\!0$, $\Gamma' = \ity{\Delta^i}{\Gamma''}$ for some $i$ s.t. $\Gamma^i = \ity{\Delta^i}{\Delta'}$, $|\Delta^i|\!>\!0$. Note that $TV(\Gamma^i)$ are pairwise disjoint, thus if $\Delta^i\!\neq\! \Gamma^i$ ($\Delta' \!\neq\! nil$) then $\cty{\Delta^i}{}$ would be closed and strictly held in $T^i$. Hence, $\Delta^i\!=\!\Gamma^i$ ($\Delta' \!=\! nil$) and similarly $\ty{\varphi_1\to\cdots\to\varphi_m}{\alpha}$ must be in $\Gamma'$, giving a non closed $\mathcal{C}$-type $T'$. If $T'=\cty{\Gamma'}{\alpha}$ then with a similar argument one has that $\Gamma'\!=\!\ity{(\omega^{\dbi{n-1}}.\ty{\varphi_1\to\cdots\to\varphi_m}{\alpha}.nil)}{\Gamma^1\!\land\cdots\land\Gamma^m}$. Therefore, $T'$ is closed iff $T$ is closed and $T'\!=\!T$. Hence, $T$ is m.c.\qedhere
\end{enumerate}
\end{proof}

\begin{definition}\label{def:completeGt}
A $\mathcal{C}$-type $T$ is called {\bf complete} if $T$ is closed, finally closed and minimally closed.
\end{definition} 
\begin{lemma}\label{lem:completeprop}
\begin{enumerate}
\item\label{comp1}
If $\cty{v.\Gamma}{\varphi}$ is complete for $v\neq\omega$ then $\cty{\Gamma}{\ty{v}{\varphi}}$ is complete.

\item\label{comp2}
$\cty{\omega.\Gamma}{\varphi}$ is complete iff $\cty{\Gamma}{\ty{\omega}{\varphi}}$ is complete.

\item\label{comp3}
$\cty{nil}{\varphi}$ is complete iff $\cty{nil}{\ty{\omega}{\varphi}}$ is complete.

\item\label{comp4}
If $\forall 1{\leq} i {\leq} m$, $T_i = \cty{\Gamma^i}{\varphi_i}$ is complete and $TV(T_i)$ are pairwise disjoint then,  for any fresh type variable $\alpha$, $T = \cty{\ity{(\omega^{\dbi{n-1}}.\ty{\varphi_1\to\cdots\to\varphi_m}{\alpha}.nil)}{\Gamma^1\!\land\cdots\land\Gamma^m}}{\alpha}$ is complete.

\end{enumerate}
\end{lemma}

\begin{proof}
\begin{enumerate}
\item By Lemmas \ref{lem:closedprop}.\ref{closed1}, \ref{lem:fclosedprop}.\ref{fc1} and \ref{lem:mclosedprop}.\ref{mc1}.

\item By Lemmas \ref{lem:closedprop}.\ref{closed1}, \ref{lem:fclosedprop}.\ref{fc1} and \ref{lem:mclosedprop}.\ref{mc2}.

\item By Lemmas \ref{lem:closedprop}.\ref{closed2}, \ref{lem:fclosedprop}.\ref{fc2} and \ref{lem:mclosedprop}.\ref{mc3}.

\item By Lemmas \ref{lem:closedprop}.\ref{closed3} and \ref{lem:mclosedprop}.\ref{mc4} one has that the $T$ described above is respectively closed and m.c. Note that $\ity{(\ty{\varphi_1\to\cdots\to\varphi_m}{\alpha})}{(\Gamma^1\!\land\cdots\land\Gamma^m)_n} \!\in\! L(T)$, thus $T$ is f.c.\qedhere 

\end{enumerate}
\end{proof}

\begin{lemma}\label{lem:inferxcomplete}
If $N$ is a $\beta$-nf then $T^{N}$ is complete.
\end{lemma}
\begin{proof}
By structural induction on $N$.
\begin{itemize}
\item Let $N \equiv \dbi{n}$. One has that $\mathtt{Infer}(N)=(\omega^{\dbi{n{-}1}}\,.\alpha.nil,\alpha)$, hence $T^{N}= \cty{\omega^{\dbi{n{-}1}}\,.\alpha.nil}{\alpha}$. Note that $L(T^N) = \{\alpha\}$. Thus, $T^{N}$ is closed and finally closed. The only two $\mathcal{C}$-types strictly held in $T^{N}$ are $\cty{\omega^{\dbi{n{-}1}}\,.\alpha.nil}{}$ and $\cty{nil}{\alpha}$ which are not closed, hence $T^N$ is minimally closed.

\item Let $N \equiv \lambda.N'$. If $(\Gamma',\varphi)  \!=\! \mathtt{Infer}(N')$ then, by IH, $T^{N'} \!=\! \cty{\Gamma'}{\varphi}$ is complete.

If $\Gamma' \!=\! v.\Gamma$ then $\mathtt{Infer}(\lambda.N') \!=\! (\Gamma,\ty{v}{\varphi})$ and $T^N \!=\! \cty{\Gamma}{\ty{v}{\varphi}}$. If $v\!\neq\! \omega$, then by Lemma \ref{lem:completeprop}.\ref{comp1} $T^N$ is complete. Otherwise, by Lemma \ref{lem:completeprop}.\ref{comp2}, $T^N$ is complete.

If $\Gamma' \!=\! nil$ then $\mathtt{Infer}(\lambda.N') \!=\! (nil,\ty{\omega}{\varphi})$ and, by Lemma \ref{lem:completeprop}.\ref{comp3}, $T^N$ is complete.

\item Let $N \equiv \dbi{n}\,N_1\cdots N_m$. If $\forall 1{\leq}i{\leq}m$, $(\Gamma^i,\varphi_{i}) \!=\! \mathtt{Infer}(N_i)$ then, by IH, $T^{N_i}$ is complete. Observe that $TV(T^{N_i})$ are pairwise disjoint because they correspond to disjoint calls of $\mathtt{Infer}$. One has that $\mathtt{Infer}(N) \!=\! (\ity{(\omega^{\dbi{n{-}1}}\,.\ty{\varphi_1\to\cdots\to\varphi_m}{\alpha}.nil)}{\Gamma^1\!\land \cdots \land \Gamma^m},\alpha)$, for some fresh type variable $\alpha$. Thus, by Lemma \ref{lem:completeprop}.\ref{comp4}, $T^N$ is complete.\qedhere

\end{itemize}
\end{proof}

Note that on items \ref{comp1} and \ref{comp4} in Lemma \ref{lem:completeprop} we only have {\it sufficiency} proofs. Following we give counterexamples for each {\it necessary} condition.

\begin{example}\label{exemp:mc}
Let $T = \cty{\Gamma}{\varphi}$ be complete. Then, for any fresh $\alpha\!\in\!\mathcal{A}$, take $T' = \cty{\Gamma}{(\ty{\alpha}{\alpha})\!\to\!\varphi}$. Therefore, $T'$ is complete but $\cty{\ty{\alpha}{\alpha}.\Gamma}{\varphi}$ is not m.c.
\end{example}

\begin{example}\label{exemp:comp}
Let $T\!=\!\cty{\beta_1 \!\to\! (\ty{\beta_2}{\beta_3}) \!\to\! \beta_4.(\ty{\beta_1}{\beta_4})\!\to\!(\ty{\beta_3}{\beta_2})\!\to\!\alpha.nil}{\alpha}$. Note that $T$ is complete but there is no such a partition of complete $\mathcal{C}$-types.
\end{example} 
Hence, to have complete $\mathcal{C}$-types which satisfy those {\it necessary} conditions, we present the notion of principal $\mathcal{C}$-types, as done in \cite{SM96a}.
\begin{definition}\label{def:principal}
Let $T$ be a complete $\mathcal{C}$-type. $T$ is called principal if:
\begin{itemize}
\item[-] $T= \cty{\omega^{\dbi{n{-}1}}\,.\alpha.nil}{\alpha}$.

\item[-] $T = \cty{nil}{\ty{\omega}{\varphi}}$ and $\cty{nil}{\varphi}$ is principal.

\item[-]  $T = \cty{\Gamma}{\ty{v}{\varphi}}$ such that either $\Gamma\neq nil$ or $v\neq\omega$ and $\cty{v.\Gamma}{\varphi}$ is principal.

\item[-] $T= \cty{\Gamma}{\alpha}$ and there are $\Gamma^1,\dots,\Gamma^m \!\in\! \mathcal{C}$ and $n\!\in\!\mathbb{N}^*$ such that $\Gamma = \ity{(\omega^{\dbi{n{-}1}}\,.\ty{\varphi_1\to\cdots\to\varphi_m}{\alpha}.nil)}{\Gamma^1\!\land\cdots\land\Gamma^m}$ and $\forall 1{\leq}i{\leq}m$, $\cty{\Gamma^i}{\varphi_i}$ is principal.

\end{itemize} 
\end{definition}

Observe that in Definition \ref{def:principal} above we explicitly require  the existence of the corresponding partition in the case $T\!=\!\cty{\Gamma}{\alpha}$ for $\Gamma \!\neq\! \omega^{\dbi{n{-}1}}\,.\alpha.nil$ and that $\cty{v.\Gamma}{\varphi}$ is also principal thus complete for $T \!=\! \cty{\Gamma}{\ty{v}{\varphi}}$ such that $\Gamma\!\neq\!nil$ or $v\!\neq\!\omega$. Although we have that, by Lemma \ref{lem:completeprop}.\ref{comp2}, $T\!=\!\cty{nil}{\ty{\omega}{\varphi}}$ is complete iff $T'\!=\!\cty{nil}{\varphi}$ is complete, this case has to be defined similarly. If in Definition \ref{def:principal} we only have instead: ``$T = \cty{nil}{\ty{\omega}{\varphi}}$'' then we would guarantee only the completeness of $T'$, letting a counterexample as in Example \ref{exemp:mc} to be presented.

\begin{lemma}\label{lem:inferxprincipal}
If $N$ is a $\beta$-nf then $T^N$ is principal.
\end{lemma}
\begin{proof}
By structural induction on $N$. By Lemma \ref{lem:inferxcomplete}, $T^N$ is complete:
\begin{itemize}
\item If $N \equiv \dbi{n}$ then $T^{N} = \cty{\omega^{\dbi{n{-}1}}\,.\alpha.nil} {\alpha}$.

\item  Let $N \equiv \lambda.N'$ and $T^{N'} \!=\!  \cty{\Gamma'}{\varphi}$. By IH $T^{N'}$ is principal. 

If $\Gamma' \!=\! v.\Gamma$ then $T^{\lambda.N'}  \!=\!  \cty{\Gamma}{\ty{v}{\varphi}}$. If $\Gamma \!=\! nil$ then,  by Lemma \ref{lem:noweak}, $v \!\neq\! \omega$. Hence, $T^{\lambda.N'} $ is principal.

Otherwise $T^{\lambda.N'} \!=\! \cty{nil}{\ty{\omega}{\varphi}}$, hence $T^{\lambda.N'}$ is principal.

\item  Let $N \equiv \dbi{n}\,N_1\cdots N_m$ and $\forall 1{\leq}i{\leq}m$, $T^{N_i}  =  \cty{\Gamma^i}{\varphi_{i}}$. Hence,  for some fresh type variable $\alpha$, $T^{N} = \cty{\ity{(\omega^{\dbi{n{-}1}}\,.\ty{\varphi_1\to\cdots\to\varphi_m}{\alpha}.nil)}{\Gamma^1\!\land \cdots \land \Gamma^m}}{\alpha}$ and, by IH, $T^{N_i}$ is principal $\forall 1{\leq}i{\leq}m$. Thus, $T^{N}$ is principal.\qedhere 

\end{itemize}
\end{proof}
Therefore, the syntactic definition of principal $\mathcal{C}$-types contains the PT for $\beta$-nfs returned by $\mathtt{Infer}$.
\begin{definition}
Let $\mathcal{P} = \{(\Gamma,\varphi) \in \mathcal{C}\!\times\!\mathcal{T}_{NF} \,|\, \cty{\Gamma}{\varphi} \; \mbox{is principal}\}$.
\end{definition}
In other words, by Lemma \ref{lem:inferxprincipal} and analogously to \cite{SM96a}: $\;Im(\mathtt{Infer}) \subseteq \mathcal{P}$

\begin{definition}
Let $FO(\alpha,\Gamma)=\{(i,\Gamma_i)\,|\,\alpha \:\mbox{is the final occurrence of}\:\, \Gamma_i, \forall 1{\leq}i{\leq}|\Gamma|\}$.
\end{definition}
The set $FO(\alpha,\Gamma)$ for $T = \cty{\Gamma}{\alpha} $ principal, specifically closed and finally closed, has properties used in the reconstruction algorithm's definition.
\begin{lemma}\label{lem:finalclosedstruct}
Let $T= \cty{\Gamma}{\alpha}$ be a $\mathcal{C}$-type. If $T$ is finally closed then $FO(\alpha,\Gamma)\neq \emptyset$. If $T$ is also closed then $FO(\alpha,\Gamma)$ has exactly one element $(i,v)$, s.t. $v = \ity{(\ty{\varphi_1\to\cdots\to\varphi_m}{\alpha})}{v'}$, for $m \geq 0$ and $\alpha \notin TV(v')$. 
\end{lemma}

\begin{proof}
Let $T \!=\! \cty{\Gamma}{\alpha}$. By Definition \ref{def:leftst}, $L(T)\!=\!\{\Gamma_i{\neq}\omega, \forall 1{\leq}i{\leq}|\Gamma|\}$, hence if $T$ is f.c. then at least one element of $\Gamma$ has $\alpha$ as its final occurrence. Let $(i,v)\!\in\!FO(\alpha,\Gamma)$. If $T$ is also closed then $\Gamma$ has exactly one positive occurrence of $\alpha$, hence $\alpha$ occurs uniquely in $v{=}\Gamma_i$. Note that $v\!\in\!\mathcal{U}_C$. If $v\!\in\!\mathcal{T}_{C}$ then by induction on its structure $v \!=\! \ty{\varphi_1\to\cdots\to\varphi_m}{\alpha}$ for $m{\geq}0$ ($v\!=\!\alpha$ if $m\!=\!0$). Otherwise, $v = \ity{v_1}{v_2}$ and $\alpha$ occurs positively either in $v_1$ or in $v_2$. Thus, by induction on the structure of elements in $\mathcal{U}_C$, commutativity and associativity of $\land$, the result holds.
\end{proof}
We introduce the algorithm $\mathtt{Recon}$, to reconstruct a $\beta$-nf $N$ from $(\Gamma,\varphi)\in \mathcal{P}$ such that $\mathtt{Infer}(N)=(\Gamma,\varphi)$, similar to the algorithm introduced in \cite{SM96a}.
\begin{definition}[Reconstruction algorithm].\\
{\small
$\mathtt{Recon}(\Gamma,\tau) = $\\

     \begin{tabular}{cl}
      {\bf Case} & $(nil,\alpha)$ \\

                 & {\bf fail}

     \end{tabular}

     \begin{tabular}{cl}
     {\bf Case} & $(\Gamma,\alpha)$ \\

                & {\bf let} $\{(i^1,u_1),\dots,(i^m,u_m)\} = FO(\alpha,\Gamma)$ \\ 

                & {\bf if} $m=1$ and $u_1 = \ity{(\ty{\tau_1\to\cdots\to\tau_n}{\alpha})}{u'}$ s.t. $\alpha\!\notin\! TV(u')$\\
              
                & \hspace{.2cm} {\bf then if} $\forall 1{\leq}i{\leq}n$ there is $\Gamma^i$ s.t. $\Gamma = \ity{\Gamma^i}{X^i}$ and $\cty{\Gamma^i}{\tau_i}$ is principal\\ 
                
                & \hspace{1.6cm} {\bf then let}  $(N_1,\Delta^1) = \mathtt{Recon}(\Gamma^1,\tau_1)$ \\
                &           \hspace{4.4cm} $\vdots$ \\
                &           \hspace{2.95cm} $(N_n,\Delta^n) = \mathtt{Recon}(\Gamma^n,\tau_n)$\\
                & \hspace{3.05cm} $\Delta' = \omega^{\dbi{i^1{-}1}}.\ty{\tau_1\to\cdots\to\tau_n}{\alpha}.nil$\\
                & \hspace{3.05cm} $\Gamma' = \ity{\Delta'}{\Gamma^1 \!\land \cdots\land \Gamma^n}$ \\

                & \hspace{3.05cm} $\Gamma = \ity{\Gamma'}{\Delta}$, s.t. $\Delta\neq\omega^{\dbi{j}}$, $\forall 1{\leq}j{\leq}|\Gamma|$ \\

                & \hspace{1.6cm} {\bf return} $(\dbi{i}^1N_1\,\cdots\,N_n, \ity{\Delta}{\Delta^1 \!\land \cdots\land \Delta^n})$ \\

                & \hspace{1.6cm} {\bf else fail} \\

                & \hspace{.2cm} {\bf else fail} 
 
     \end{tabular}

     \begin{tabular}{cl}
     {\bf Case} & $(\Gamma,\ty{u}{\tau})$ \\

                & {\bf if} $\Gamma=nil$ and $u=\omega$\\

                & \hspace{.2cm} {\bf then let} $(N,\Delta) = \mathtt{Recon}(nil,\tau)$\\

                & \hspace{.2cm} {\bf else let} $(N,\Delta) = \mathtt{Recon}(u.\Gamma,\tau)$ \\
               
                & {\bf if} $\Delta = nil$\\  

                & \hspace{.2cm} {\bf then return} $(\lambda.N,\Delta)$\\

                & \hspace{.2cm} {\bf else fail}
     \end{tabular}
}
\end{definition} 

\begin{lemma}\label{lem:PxInfer}
Let $(\Gamma,\varphi)\in\mathcal{P}$. Then $\mathtt{Recon}(\Gamma,\varphi)=(N,nil)$,  $N$ a $\beta$-nf such that $\mathtt{Infer}(N)=(\Gamma,\varphi)$.
\end{lemma}
\begin{proof}
By recurrence on the number of calls to $\mathtt{Recon}$.
\begin{itemize}
\item Case $(\Gamma,\alpha)$. Let $T=\cty{\Gamma}{\alpha}$.

By hypothesis $(\Gamma,\alpha)\in\mathcal{P}$, thus $T$ is principal and in particular closed and f.c.. By Lemma \ref{lem:finalclosedstruct}, $FO(\alpha,\Gamma)\!=\!\{(i,\ity{(\ty{\varphi_1\to\cdots\to\varphi_m}{\alpha})}{v'})\}$ where $\alpha\notin TV(v')$. Since $\Gamma_i$ is the only occurrence of $\alpha$ in $\Gamma$, $\Gamma \!=\! \ity{(\omega^{\dbi{i{-}1}}.\ty{\varphi_1\to\cdots\to\varphi_m}{\alpha}.nil)}{\Delta''}$ s.t. $\alpha\!\notin\!TV(\Delta'')$.

If $m \!=\! 0$, then in $\mathtt{Recon}$ one has $\Gamma'\!=\! \Delta' \!=\!\omega^{\dbi{i{-}1}}.\alpha.nil$, hence $T \!=\! \cty{\ity{\Gamma'}{\Delta''}}{\alpha}$. $T$ is m.c., thus $\Delta'' = nil$ and $\Gamma = \Gamma'$. Then, $\mathtt{Recon}(\Gamma,\alpha) = (\dbi{i}\,,nil)$ and $\mathtt{Infer}(\dbi{i}\,) = (\omega^{\dbi{i{-}1}}\,.\alpha.nil,\alpha)$.

Otherwise, there are $\Gamma^1,\dots,\Gamma^m$ and $n\in\mathbb{N}^*$ s.t. $\Gamma = \ity{(\omega^{\dbi{n{-}1}}\,.\ty{\varphi_1\to\cdots\to\varphi_m}{\alpha}.nil)}{\Gamma^1\!\land\cdots\land\Gamma^m}$ and $\forall 1{\leq}j{\leq}m$, $\cty{\Gamma^j}{\varphi_j}$ is principal. Hence, $n=i$ and by IH $\forall 1{\leq}j{\leq}m$, $\mathtt{Recon}(\Gamma^j,\varphi_j)=(N_j,nil)$, $N_j$ a $\beta$-nf s.t. $\mathtt{Infer}(N_j) = (\Gamma^j,\varphi_j)$. Hence in $\mathtt{Recon}$ one has that $\Gamma \!=\! \Gamma'$, consequently $\Delta \!=\! nil$. Then, $\mathtt{Recon}(\Gamma,\alpha) \!=\! (\dbi{i}\,N_1\,\cdots\,N_m,nil)$ and $\mathtt{Infer}(\dbi{i}\,N_1\cdots N_m) \!=\! (\ity{(\omega^{\dbi{i{-}1}}\,.\ty{\varphi_1\to\cdots\to\varphi_m}{\alpha}.nil)}{\Gamma^1\!\land\cdots\land\Gamma^m},\alpha)$.

\item Case $(\Gamma,\ty{v}{\varphi})$. Let $T= \cty{\Gamma}{\ty{v}{\varphi}}$.

By hypothesis  $(\Gamma,\ty{v}{\varphi})\in\mathcal{P}$, thus  $T$ is principal. 

If $\Gamma=nil$ and $v=\omega$ then $T'= \cty{nil}{\varphi}$ is principal and, by IH, $\mathtt{Recon}(nil,\varphi)=(N,nil)$, $N$ a $\beta$-nf s.t. $\mathtt{Infer}(N)=(nil,\varphi)$.  Thus, $\mathtt{Recon}(nil,\ty{\omega}{\varphi})=(\lambda.N,nil)$ and $\mathtt{Infer}(\lambda.N) = (nil,\ty{\omega}{\varphi})$. 

Otherwise, $T' \!=\! \cty{v.\Gamma}{\varphi}$ is principal. By IH, $\mathtt{Recon}(v.\Gamma,\varphi)\!=\!(N,nil)$, $N$ a $\beta$-nf s.t. $\mathtt{Infer}(N)\! =\! (v.\Gamma,\varphi)$. Hence, $\mathtt{Recon}(\Gamma,\ty{v}{\varphi})=(\lambda.N,nil)$ and $\mathtt{Infer}(\lambda.N) = (\Gamma,\ty{v}{\varphi})$.\qedhere 
\end{itemize}
\end{proof}
Observe that, by Lemma \ref{lem:PxInfer}, we have that:
$\mathcal{P} \subseteq Im(\mathtt{Infer})$.
Thus, $\mathcal{P}$ is the set of all, and only, principal typings for $\beta$-nfs in $SM_r$. Therefore,
$\mathcal{P} = Im(\mathtt{Infer})$.

\section{Conclusion}

In this paper, we introduced the first intersection type system in de
Bruijn indices for which the principle typings property for $\beta$-normal forms holds.

The restriction in the system of \cite{SM96a} prevents both that
system and our own system presented here, from having SR in the usual
sense. This is not the case however for the system of
\cite{VAK2008}. However, every $\beta$-nf is typeable in the
introduced system, as in the one in \cite{SM96a}, a property that does
not hold for the simply typed system. We then prove the PT property
for $\beta$-nfs and a characterisation of PT is given. This de Bruijn
version of the typing system in \cite{SM96a} was introduced as a first
step towards some extended systems in which PT depends on more complex
syntactic operations such as expansion \cite{new}. 

As future work, we will introduce a de Bruijn version for systems such as the ones in \cite{CDV80} and \cite{roc84} and try to add similar systems to both $\lambda\sigma$ and $\lambda s_e$. There are works on intersection types and explicit substitution, e.g. \cite{LLDDvB}, but no work for systems where the composition of substitutions is allowed.

\bibliographystyle{eptcs}

\begin{thebibliography}{1}

\bibitem{ACCL91}
M.~Abadi, L.~Cardelli, P.-L. Curien and J.-J. L\'evy (1991):
\newblock \emph{Explicit Substitutions}.
\newblock {\sl J. func. program.}, 1(4):375--416.

\bibitem{ARKa2001a}
M.~Ayala-Rinc\'on and F.~Kamareddine (2001):
\newblock \emph{Unification via the $\lambda s_e$-Style of Explicit Substitution}.
\newblock {\sl Logical journal of the IGPL}, 9(4):489--523.

\bibitem{bakel95}
S. van Bakel (1995): 
\newblock \emph{Intersection Type Assignment Systems}. 
\newblock {\sl Theoret. comput. sci.}, 151:385-435.

\bibitem{BarCDC1983}
H.~Barendregt, M.~Coppo and M.~Dezani-Ciancaglini (1983):
\newblock \emph{A filter lambda model and the completeness of type assignment}.
\newblock {\sl J. symbolic logic}, 48:931--940.

\bibitem{Bar1984}
H.~Barendregt (1984):
\newblock {\sl The Lambda Calculus: Its Syntax and
Semantics}. North-Holland.

\bibitem{dB72}
N.G. de~Bruijn (1972):
\newblock \emph{Lambda-Calculus Notation with Nameless Dummies, a Tool for Automatic
  Formula Manipulation, with Application to the Church-Rosser Theorem}.
\newblock {\sl Indag. Mat.}, 34(5):381--392.

\bibitem{dB78}
N.G. de~Bruijn (1978):
\newblock \emph{A namefree lambda calculus with facilities for internal definition
  of expressions and segments}.
\newblock {\sl T.H.-Report 78-WSK-03}, Technische Hogeschool Eindhoven, Nederland.

\bibitem{CarWe2004}
S.~Carlier and J.~B. Wells (2004):
\newblock \emph{Type Inference with Expansion Variables and Intersection Types in
  System E and an Exact Correspondence with $\beta$-reduction}.
\newblock In {\sl  Proc. of PPDP '04}, pp. 132--143. ACM.

\bibitem{CarWeITRS04}
S.~Carlier and J.~B. Wells (2005):
\newblock \emph{Expansion: the Crucial Mechanism for Type Inference with Intersection Types: a Survey and Explanation}.
\newblock In {\sl Proc. of ITRS '04}, ENTCS 136:173--202. Elsevier.

\bibitem{CDC1978}
M.~Coppo and M.~Dezani-Ciancaglini (1978):
\newblock \emph{A new type assignment for lambda-terms}.
\newblock {\sl Archiv f\"ur mathematische logik},
  19:139--156.

\bibitem{CDC1980}
M.~Coppo and M.~Dezani-Ciancaglini (1980):
\newblock \emph{An Extension of the Basic Functionality Theory for the
  $\lambda$-Calculus}.
\newblock {\sl Notre dame j. formal logic}, 21(4):685--693.

\bibitem{CDV80}
M. Coppo, M. Dezani-Ciancaglini and B.~Venneri (1980):
\newblock \emph{Principal Type Schemes and $\lambda$-calculus Semantics}.
In J.P.~Seldin and J.R.~Hindley (eds), {\sl To H.B.~Curry: Essays on combinatory logic, lambda calculus and formalism}, pp. 536--560. Academic Press.

\bibitem{curfe}
H.~B. Curry and R.~Feys (1958):
\newblock {\sl Combinatory Logic}, vol.~1.
\newblock {North Holland}.

\bibitem{DG94}
F.~Damiani and P.~Giannini (1994):
\newblock \emph{A Decidable Intersection Type System based on Relevance}. 
\newblock In {\sl Proc. of TACS’94}, { LNCS} 789:707–725. Springer-Verlag. 

\bibitem{Hi97} 
J.~R. Hindley (1997): 
\newblock {\em Basic Simple Type Theory}.  
\newblock {\sl Cambridge Tracts in Theoretical Computer Science}, 42.  
Cambridge University Press.

\bibitem{KaNo2007}
F.~Kamareddine and K.~Nour (2007):
\newblock \emph{A completeness result for a realisability semantics for an
  intersection type system}.
\newblock {\sl Annals pure and appl. logic}, 146:180--198.

\bibitem{new}
F.~Kamareddine, K.~Nour, Vincent Rahli and J.B.~Wells (2009):
\newblock \emph{On Realisability Semantics for Intersection Types with
Expansion variables}.
\newblock {Submitted for Publication}.

\bibitem{fayo}
F.~Kamareddine and A.~R\'{\i}os (1995):
\newblock \emph{A $\lambda$-calculus \`a la de Bruijn with Explicit Substitutions}.
\newblock In {\sl Proc. of PLILP'95}, { LNCS} 982:45--62. Springer.

\bibitem{DBLP:journals/cj/KamareddineR02}
{Fairouz Kamareddine and
               Alejandro R\'{\i}os} (2002):
  \newblock{Pure Type Systems with de Bruijn Indices},
  \newblock{Computer Journal} 45(2): 187--201.

\bibitem{KW2004}
A.J. Kfoury and J.B. Wells (2004):
\newblock \emph{Principality and type inference for intersection types using expansion variables}. 
\newblock {\sl Theoret. comput. sci.}, 311(1--3):1--70.

\bibitem{Kri93}
J-L.~Krivine (1993):
\newblock {\sl Lambda-calculus, types and models}.
\newblock Ellis Horwood.

\bibitem{LLDDvB} S.~Lengrand, P.~Lescanne, D.~Dougherty, M.~Dezani-Ciancaglini and S.~van Bakel (2004): 
\newblock \emph{Intersection types for explicit substitutions}. 
\newblock{ {\sl Inform. and comput.}, 189(1):17–-42}.

\bibitem{Milner78}
R.~Milner (1978):
\newblock \emph{A theory of type polymorphism in programming}.
\newblock {\sl J. comput. and system sci.}, 17(3):348--375.

\bibitem{NGDV}
R.~P. Nederpelt, J.~H. Geuvers and R.~C. de~Vrijer (1994):
\newblock {\sl Selected papers on Automath}.
\newblock {North-Holland}.

\bibitem{Po1980}
G.~Pottinger (1980):
\newblock \emph{A type assignment for the strongly normalizable $\lambda$-terms}.
\newblock In J.P. Seldin and J.~R. Hindley (eds), {\sl To H. B. Curry:
  Essays on combinatory logic, lambda calculus and formalism}, pp. 561--578.
  Academic Press.

\bibitem{roc84} S.~Ronchi Della Rocca and B.~Venneri (1984):
\newblock  \emph{Principal Type Scheme for an Extended Type Theory}. 
\newblock {\sl Theoret. comput. sci.}, 28:151--169.

\bibitem{roc88} S.~ Ronchi Della Rocca (1988):
\newblock \emph{Principal Type Scheme and Unification for Intersection Type Discipline}.
\newblock {\sl Theoret. comput. sci.}, 59:181--209.

\bibitem{SM96a}
E.~Sayag and M.~Mauny (1996):
\newblock \emph{Characterization of principal type of normal forms in intersection type system}.
\newblock In {\sl Proc. of FSTTCS'96}, LNCS, 1180:335--346. Springer.

\bibitem{SM96b}
E.~Sayag and M.~Mauny (1996):
\newblock \emph{A new presentation of the intersection type discipline through principal typings of normal forms}.
\newblock {\sl Tech. rep. RR-2998}, INRIA.

\bibitem{SM97}
E.~Sayag and M.~Mauny (1997):
\newblock \emph{Structural properties of intersection types}.
\newblock In {\sl Proc. of LIRA'97}, pp. 167–-175. Novi Sad, Yugoslavia.

\bibitem{VAK2008} 
D.~Ventura, M.~Ayala-Rinc\'on and F.~Kamareddine (2009):
\newblock \emph{Intersection Type System with de Bruijn Indices}.
\newblock Available at \url{http://www.mat.unb.br/ventura/papers/longSLALM2008.pdf} - revised version to appear in {\sl The many sides of logic}. Studies in logic, College publications. London. 

\bibitem{We2002}
J.B. Wells (2002):
\newblock \emph{The essence of principal typings}.
\newblock In {\sl Proc. of ICALP 2002}, LNCS, 2380:913--925. Springer.

\end{thebibliography}

\end{document}